\RequirePackage{amsmath}
\documentclass[11pt,envcountsame]{llncs}
\usepackage{amssymb}
\usepackage{color,graphicx,setspace}
\usepackage{geometry}
\usepackage{algorithm}
\usepackage[noend]{algpseudocode}

\newcommand\SSS{\mathcal S}

\title{Faster Minimization of Tardy Processing\newline Time on a Single Machine
}

\def\lastandname{\unskip,\\and}
\author{Karl Bringmann\inst{1}\thanks{This work is part of the project TIPEA that has received funding from the European
Research Council (ERC) under the European Union’s Horizon 2020 research and innovation
programme (grant agreement No. 850979).}\and
    Nick Fischer\inst{1}${}^\star$
    \and
    Danny Hermelin\inst{3} \and
    Dvir Shabtay\inst{3}\and
    Philip Wellnitz\inst{2}}
\institute{
    Saarland University and Max Planck Institute for Informatics, Saarbr\"ucken, Germany\\
    \email{bringmann@cs.uni-saarland.de}, \email{nfischer@mpi-inf.mpg.de}\\[2ex]
    \and
    Max Planck Institute for Informatics, Saarbr\"ucken, Germany\\
    \email{wellnitz@mpi-inf.mpg.de}\\[2ex]
	\and
    Department of Industrial Engineering and Management,\\Ben-Gurion University of the Negev, Beersheba, Israel\\
    \email{\{hermelin,dvirs\}@bgu.ac.il}
}

\authorrunning{K. Bringmann, N. Fischer, D. Hermelin, D. Shabtay, and P. Wellnitz}

\begin{document}
\maketitle

\begin{abstract}
    This paper is concerned with the $1||\sum p_jU_j$ problem, the problem  of minimizing the total processing time of tardy jobs on a single machine. This is not only a fundamental scheduling problem, but also a very important problem from a theoretical point of view as it generalizes the Subset Sum problem and is closely related to the 0/1-Knapsack problem. The problem is well-known to be NP-hard, but only in a weak sense, meaning it admits pseudo-polynomial time algorithms. The fastest known pseudo-polynomial time algorithm for the problem is the famous Lawler and Moore algorithm which runs in $O(P \cdot n)$ time, where $P$ is the total processing time of all $n$ jobs in the input. This algorithm has been developed in the late 60s, and has yet to be improved to date.

    In this paper we develop two new algorithms for $1||\sum p_jU_j$, each improving on Lawler and Moore's algorithm in a different scenario:
    \begin{itemize}
        \item Our first algorithm runs in $\tilde{O}(P^{7/4})$ time\footnote{Throughout the paper we use $\tilde{O}(\cdot)$ to suppress logarithmic factors.}, and outperforms Lawler and Moore's algorithm in instances where $n=\tilde{\omega}(P^{3/4})$.
        \item Our second algorithm runs in $\tilde{O}(\min\{P \cdot D_{\#}, P + D\})$ time, where $D_{\#}$ is the number of \emph{different} due dates in the instance, and $D$ is the sum of all different due dates. This algorithm improves on Lawler and Moore's algorithm when $n=\tilde{\omega}(D_{\#})$ or $n=\tilde{\omega}(D/P)$. Further, it extends the known $\tilde{O}(P)$ algorithm for the single due date special case of $1||\sum p_jU_j$ in a natural way.
    \end{itemize}
    \smallskip
    Both algorithms rely on basic primitive operations between sets of integers and vectors of integers for the speedup in their running times. The second algorithm relies on fast polynomial multiplication as its main engine, while for the first algorithm we define a new ``skewed'' version of $(\max,\min)$-convolution which is interesting in its own right.

\end{abstract}
\section{Introduction}

In this paper we consider the problem of minimizing the total processing times of tardy jobs on a single machine. In this problem we are given a set of $n$ jobs $J=\{1,\ldots,n\}$, where each job~$j$ has a \emph{processing time} $p_j \in \mathbb{N}$ and a \emph{due date} $d_j \in \mathbb{N}$. A \emph{schedule} $\sigma$ for $J$ is a permutation $\sigma: \{1,\ldots,n\} \to \{1,\ldots,n\}$. In a given schedule $\sigma$, the \emph{completion time} $C_j$ of a job $j$ under $\sigma$ is given by $C_j = \sum_{\sigma(i) \leq \sigma(j)} p_i$, that is, the total processing time of jobs preceding $j$ in $\sigma$ (including $j$ itself). Job $j$ is \emph{tardy} in $\sigma$ if $C_j > d_j$, and \emph{early} otherwise. Our goal is find a schedule with minimum total processing time of tardy jobs. If we assign a binary indicator variable $U_j$ to each job $j$, where $U_j=1$ if $j$ is tardy and otherwise $U_j=0$, our objective function can be written as $\sum p_j U_j$. In the standard three field notation for scheduling problems of Graham~\cite{Graham1969}, this problem is denoted as the $1||\sum p_j U_j$ problem (the 1 in the first field indicates a single machine model, and the empty second field indicates there are no additional constraints).

The $1|| \sum p_j U_j$ problem is a very natural and fundamental scheduling problem, which models a very basic scheduling scenario. As it includes Subset Sum as a special case (see below), the $1|| \sum p_j U_j$ problem is NP-hard. However, it is only hard in the weak sense, meaning it admits pseudo-polynomial time algorithms. The focus of this paper is on developing fast pseudo-polynomial time algorithms for $1|| \sum p_j U_j$, improving in several settings on the best previously known solution from the late 60s. Before we describe our results, we discuss the previously known state of the art of the problem, and describe how our results fit into this line of research.

\subsection{State of the Art}

A famous generalization of the $1|| \sum p_jU_j$ problem is the $1|| \sum w_j U_j$ problem. Here, each job $j$ also has a weight $w_j$ in addition to its processing time $p_j$ and due date $d_j$, and the goal is to minimize the total weight (as opposed to total processing times) of tardy jobs. This problem has already been studied in the 60s, and even appeared in Karp's fundamental paper from 1972~\cite{Karp72}. The classical algorithm of Lawler and Moore~\cite{LawlerMoore} for the problem is one of the earliest and most prominent examples of pseudo-polynomial algorithms, and it is to date the fastest known algorithm even for the special case of $1|| \sum p_jU_j$. Letting $P= \sum_{j \in J} p_j$, their result can be stated as follows:
\begin{theorem}[\cite{LawlerMoore}]
    \label{thm:LawlerMoore}%
    $1|| \sum w_jU_j$ and $1|| \sum p_jU_j$ can both be solved in $O(P \cdot n)$ time.
\end{theorem}

Note that as we assume that all processing times are integers, we have $n \leq P$, and so the running time of the algorithm in Theorem~\ref{thm:LawlerMoore} can be bounded by $O(P^2)$. In fact, it makes perfect sense to analyze the time complexity of a pseudo-polynomial time algorithm for either problems in terms of $P$, as~$P$ directly corresponds to the total input length when integers are encoded in unary. Observe that while the case of $n=P$ (all jobs have unit processing times) essentially reduces to sorting, there are several non-trivial cases where $n$ is smaller than $P$ yet still quite significant in the $O(P \cdot n)$ term of Theorem~\ref{thm:LawlerMoore}. The fundamental question this paper addresses is:
\begin{quote}
    \centering
    ``\emph{Can we obtain algorithms with running times $O(P^{2-\varepsilon})$, for any fixed $\varepsilon > 0$,\\for either $1|| \sum w_jU_j$ or $1|| \sum p_jU_j$} ?''
\end{quote}

For $1|| \sum w_jU_j$ there is some evidence that the answer to this question should be no. Karp~\cite{Karp72} observed that the special case of the $1||\sum w_jU_j$ problem where all jobs have the same due date $d$, the $1|d_j=d|\sum w_jU_j$ problem, is essentially equivalent to the classical 0/1-Knapsack problem. Cygan \emph{et al.}~\cite{CyganMWW19} and K{\"{u}}nnemann \emph{et al.}~\cite{KunnemannPS17} studied the $(\min,+)$-Convolution problem (see Section~\ref{section:preliminaries}), and conjectured that the $(\min,+)$-convolution between two vectors of length $n$ cannot be computed in $\tilde{O}(n^{2-\varepsilon})$ time, for any $\varepsilon > 0$. Under this $(\min,+)$-Convolution Conjecture, they obtained lower bounds for several Knapsack related problems. In our terms, their result can be stated as follows:
\begin{theorem}[\cite{CyganMWW19,KunnemannPS17}]
    \label{thm:Known LB1}%
    There is no $\tilde{O}(P^{2-\varepsilon})$ time algorithm for the $1|d_j=d| \sum w_jU_j$ problem, for any $\varepsilon > 0$, unless the $(\min,+)$-Convolution Conjecture is false. In particular, $1|| \sum w_jU_j$ has no such algorithm under this conjecture.
\end{theorem}

Analogous to the situation with $1||\sum w_jU_j$, the special case of $1||\sum p_j U_j$ where all jobs have the same due date $d$ (the $1|d_j=d|\sum p_j U_j$ problem) is equivalent to the classical Subset Sum problem. Recently, there has been significant improvements for Subset Sum resulting in algorithms with $\tilde{O}(T+n)$ running times~\cite{Bring17,KoiliarisX17}, where $n$ is number of integers in the instance and $T$ is the target. Due to the equivalence between the two problems, this yields the following result for the $1|d_j=d| \sum p_jU_j$ problem:
\begin{theorem}[\cite{Bring17,KoiliarisX17}]
    \label{thm:SubsetSumAlg}%
    $1|d_j=d| \sum p_jU_j$ can be solved in $\tilde{O}(P)$ time.
\end{theorem}

On the other hand, due to equivalence of $1|d_j=d|\sum p_jU_j$ and Subset Sum, we also know that Theorem~\ref{thm:SubsetSumAlg} above cannot be significantly improved unless the Strong Exponential Time Hypothesis (SETH) fails. Specifically, combining a recent reduction from $k$-SAT to Subset Sum~\cite{AbboudBHS17} with the equivalence of Subset Sum and $1|d_j=d|\sum p_jU_j$, yields the following:
\begin{theorem}[\cite{AbboudBHS17}]
    \label{thm:Known LB2}%
    There is no $\tilde{O}(P^{1-\varepsilon})$ time algorithm for the $1|d_j=d| \sum p_jU_j$ problem, for any~$\varepsilon > 0$, unless SETH fails.
\end{theorem}

Nevertheless, Theorem~\ref{thm:Known LB2} still leaves quite a big gap for the true time complexity of $1||\sum p_jU_j$, as it can potentially be anywhere between the $\tilde{O}(P)$ time known already for the special case of $1|d_j=d|\sum p_jU_j$ (Theorem~\ref{thm:SubsetSumAlg}), and the $O(Pn)=O(P^2)$ time of Lawler and Moore's algorithm (Theorem~\ref{thm:LawlerMoore}). This is the starting point of our paper.

\subsection{Our Results}


The main contribution of this paper is two new pseudo-polynomial time algorithms for $1||\sum p_jU_j$, each improving on Lawler and Moore's algorithm in a different sense. Our algorithms take a different approach to that of Lawler and Moore in that they rely on fast operators between sets and vectors of numbers.

Our first algorithm improves Theorem~\ref{thm:LawlerMoore} in case there are sufficiently many jobs in the instance compared to the total processing time. More precisely, our algorithm has a running time of $\tilde{O}(P^{7/4})$, and so it is faster than Lawler and Moore's algorithm in case $n=\tilde{\omega}(P^{3/4})$.
\begin{theorem}
    \label{thm:ConvScheduler}%
    $1|| \sum p_jU_j$ can be solved in $\tilde{O}(P^{7/4})$ time.
\end{theorem}
The algorithm in Theorem~\ref{thm:ConvScheduler} uses a new kind of convolution which we
coined ``Skewed Convolution'' and is interesting in its own right. In fact, one of the main technical contributions of this paper is a fast algorithm for the  $(\max,\min)$-Skewed-Convolution problem (see definition in Section~\ref{section:preliminaries}).

Our second algorithm for $1|| \sum p_jU_j$ improves Theorem~\ref{thm:LawlerMoore} in case there are not too many different due dates in the problem instance; that is, $D_{\#}=|\{d_j : j \in J\}|$ is relatively small when compared to $n$. This is actually a very natural assumption, for instance in cases where delivery costs are high and products are batched to only few shipments. Let~$D$ denote the sum of the different due dates in our instance. Then our second result can be stated as follows:
\begin{theorem}
    \label{thm:SumsetScheduler}
    $1|| \sum p_jU_j$ can all be solved in $\tilde{O}(\min\{P \cdot D_{\#}, P + D\})$ time.
\end{theorem}
The algorithm in Theorem~\ref{thm:SumsetScheduler} uses basic operations between sets of numbers, such as the sumset operation (see Section~\ref{section:preliminaries}) as basic primitives for its computation, and ultimately relies on fast polynomial multiplication for its speedup. It should be noted that Theorem~\ref{thm:SumsetScheduler} includes the $\tilde{O}(P)$ result of Theorem~\ref{thm:SubsetSumAlg} for $1|d_j=d|\sum p_jU_j$ as a special case where $D_{\#} = 1$ or $D=d$.

\subsection{Roadmap}

The paper is organized as follows. In Section~\ref{section:preliminaries} we discuss all the basic primitives that are used by our algorithms, including some basic properties that are essential for the algorithms. We then present our second algorithm in Section~\ref{section:SumsetScheduler}, followed by our first algorithm in Section~\ref{section:ConvScheduler}. Section~\ref{section:SkewedConv} describes our fast algorithm for the skewed version of $(\max,\min)$-convolution, and is the main technical part of the paper. Finally, we conclude with some remarks and open problems in Section~\ref{section:discussion}.

\section{Preliminaries}
\label{section:preliminaries}%

In the following we discuss the basic primitives and binary operators between sets/vectors of integers that will be used in our algorithms. In general, we will use the letters $X$ and $Y$ to denote sets of non-negative integers (where order is irrelevant), and the letters $A$ and $B$ to denote vectors of non-negative integers.

\paragraph*{Sumsets}

The most basic operation used in our algorithms is computing the sumset of two sets of
non-negative integers:
\begin{definition}[Sumsets]
    \label{def:sumset}%
    Given two sets of non-negative integers $X_1$ and $X_2$, the \emph{sumset} of $X_1$ and~$X_2$, denoted $X_1 \oplus X_2$, is defined by
    \[
    X_1 \oplus X_2 = \{ x_1+x_2 : x_1 \in X_1, x_2 \in X_2\}.
\]
\end{definition}

Clearly, the sumset $X_1 \oplus X_2$ can be computed in $O(|X_1| \cdot |X_2|)$ time. However, in certain cases we can do better using fast polynomial multiplication. Consider the two polynomials $p_1[\alpha] = \sum_{x \in X_1} \alpha^x$ and $p_2[\beta] = \sum_{x \in X_2} \beta^x$. Then the exponents of all terms in $p_1 \cdot p_2$ with non-zero coefficients correspond to elements in the sumset $X_1 \oplus X_2$. Since multiplying two polynomials of maximum degree $d$ can be done in $O(d \log d)$ time~\cite{CormenLRC2009}, we have the following:
\begin{lemma}
    \label{lem:sumset}%
    Given two sets of non-negative integers $X_1,X_2 \subseteq \{0,\ldots,P\}$, one can compute the sumset $X_1 \oplus X_2$ in $O(P \log P)$ time.
\end{lemma}

\paragraph*{Set of all Subset Sums}

Given set of non-negative integers $X$, we will frequently be using the set of all sums generated by subsets of $X$:
\begin{definition}[Subset Sums]
    \label{def:SubsetSumSet}%
    For a given set of non-negative integers $X$, define the set of all subset sums $\SSS(X)$
    as the set of integers given by
    \[
    \SSS(X)= \Big\{\sum_{x \in Y} x : Y \subseteq X\Big\}.
\]
    Here, we always assume that $0 \in \SSS(X)$ (as it is the sum of the empty set).
\end{definition}

We can use Lemma~\ref{lem:sumset} above to compute $\SSS(X)$ from $X$ rather efficiently: First, split $X$ into two sets $X_1$ and $X_2$ of roughly equal size. Then recursively compute $\SSS(X_1)$ and $\SSS(X_2)$. Finally, compute $\SSS(X)=\SSS(X_1) \oplus \SSS(X_2)$ via Lemma~\ref{lem:sumset}. The entire algorithm runs in $\tilde{O}(\sum_{x \in X} x)$ time.

\begin{lemma}[\cite{KoiliarisX17}]
    \label{lem:SubsetSumSet}%
    Given a set of non-negative integers $X$, with $P= \sum_{x \in X} x$,  one can compute $\SSS(X)$ in $\tilde{O}(P)$ time.
\end{lemma}




\paragraph*{Convolutions}

Given two vectors $A = (A[i])^n_{i=0}$, $B = (B[j])^n_{j=0}$, the $(\circ,\bullet)$-Convolution problem for binary operators $\circ$ and $\bullet$ is to compute a vector $C = (C[k])^{2n}_{k=0}$ with
\[
C[k] =\bigcirc_{i+j=k} A[i] \bullet B[j].
\]
A prominent example of a convolution problem is $(\min,+)$-Convolution discussed above; another similarly prominent example is $(\max,\min)$-Convolution which can be solved in  $\tilde{O}(n^{3/2})$ time~\cite{Kosaraju1989}. For our purposes, it is convenient to look at a \emph{skewed} variant of this problem:
\begin{definition}[Skewed Convolution]
    \label{def:SkewedConv}%
    Given two vectors $A = (A[i])^n_{i=0}$, $B = (B[j])^n_{j=0}$, we define the $(\max,\min)$-Skewed-Convolution problem to be the problem of computing the vector $C= (C[k])^{2n}_{k=0}$ where the $k$th entry in $C$ equals
    \[
    C[k] =\max_{i+j=k} \min \{A[i],B[j]+k\}
\]
    for each $k \in \{0,\ldots,2n\}$.
\end{definition}

The main technical result of this paper is an algorithm for $(\max,\min)$-Skewed-Convolution that is significantly faster than the naive $O(n^2)$ time algorithm.
\begin{theorem}
    \label{thm:SkewedConv}%
    The $(\max,\min)$-Skewed-Convolution problem for vectors of length~$n$ can be solved in $\tilde{O}(n^{7/4})$ time.
\end{theorem}

\section{Algorithm via Sumsets and Subset Sums}
\label{section:SumsetScheduler}%


In the following section, we provide a proof of Theorem~\ref{thm:SumsetScheduler} by presenting an algorithm for $1||\sum p_jU_j$ running in $\tilde{O}(\min\{P \cdot D_{\#}, P + D\})$ time. Recall that $J=\{1,\ldots,n\}$ denotes our input set of jobs, and $p_j$ and $d_j$ respectively denote the processing time and due date of job $j \in \{1,\ldots,n\}$. Our goal is to determine the minimum total processing time of tardy jobs in any schedule for $J$. Throughout the section we let $d^{(1)}<\cdots<d^{(D_\#)}$ denote the $D_\# \leq n$ \emph{different} due dates of the jobs in $J$.

A key observation for the $1||\sum p_jU_j$ problem, used already by Lawler and Moore, is that any instance of the problem always has an optimal schedule of a specific type, namely an Earliest Due Date schedule.
An Earliest Due Date (EDD) schedule is a schedule $\pi: J \to \{1,\ldots,n\}$ such that
\begin{itemize}
    \item any early job precedes all late jobs in $\pi$, and
    \item any early job precedes all early jobs with later due dates.
\end{itemize}
In other words, in an EDD schedule all early jobs are scheduled before all tardy jobs, and all early jobs are scheduled in non-decreasing order of due dates.
\begin{lemma}[\cite{LawlerMoore}]
    \label{lem:EDD}%
    Any $1||\sum p_jU_j$ instance has an optimal schedule which is EDD.
\end{lemma}

The $D_\#$-many due dates in our instance partition the input set of job $J$ in a natural manner: Define $J_i=\{j : d_j = d^{(i)}\}$ for each $i \in \{1,\ldots,D_{\#}\}$. Furthermore, let $X_i = \{ p_j : j \in J_i \}$ the processing-times of job in $J_i$. According to Lemma~\ref{lem:EDD} above, we can restrict our attention to EDD schedules. Constructing such a schedule corresponds to choosing a subset $E_i \subseteq J_i$ for each due date $d^{(i)}$ such that $\sum_{j \in E_\ell, \ell \leq i} p_j \leq d^{(i)}$ holds for each $i \in \{1,\ldots,D_\#\}$. Moreover, the optimal EDD schedule maximizes the total sum of processing times in all selected $E_i$'s.

Our algorithm is given in Algorithm~\ref{alg:sumset-scheduler}. It successively computes sets $S_1,\ldots,S_{D_{\#}}$, where set $S_i$ corresponds to the set of jobs $J_1 \cup \cdots \cup J_i$. In particular, $S_i$ includes the total processing-time of any possible set-family of early jobs $\{E_1, \ldots, E_i\}$. Thus, each $x \in S_i$ corresponds to the total processing time of early jobs in a subset of $J_1 \cup \cdots \cup J_i$. The maximum value $x \in S_{D_{\#}}$ therefore corresponds to the maximum total processing time of early jobs in any schedule for $J$. Thus, the algorithm terminates by returning the optimal total weight of tardy jobs $P-x$.

\begin{algorithm}[h]
    \caption{\textsc{SumsetScheduler($J$)}} \label{alg:sumset-scheduler}
    \begin{algorithmic}[1]
        \State Let $d^{(1)} < \ldots < d^{(D_{\#})}$ denote the different due dates of jobs in $J$.
        \State Compute $X_i = \{p_j : d_j = d^{(i)}\}$ for each $i \in \{1,\ldots,D_{\#}\}$.
        \State Compute $\SSS(X_1), \ldots, \SSS(X_{D_{\#}})$.
        \State Let $S_0=\emptyset$.
        \For{$i=1,\ldots,D_\#$}
        \Statex\textbf{\,\,--\,\,\,}Compute $S_i=S_{i-1} \oplus \SSS(X_i)$.
        \Statex\textbf{\,\,--\,\,\,}Remove any $x \in S_i$ with $x > d^{(i)}$.
        \EndFor
        \State Return $P-x$, where $x$ is the maximum value in $S_{D_{\#}}$.
    \end{algorithmic}
\end{algorithm}

Correctness of our algorithm follows immediately from the definitions of sumsets and subset sums, and from the fact that we prune out elements $x \in S_i$ with $x > d^{(i)}$ at each step of the algorithm. This is stated more formally in the lemma below.

\begin{lemma}
    Let $i\in\{1,\ldots,D_\#\}$, and let $S_i$ be the set of integers at the end of the second step of $5(i)$. Then $x \in S_i$ if and only if there are sets of jobs $E_1 \subseteq J_1, \ldots, E_i \subseteq J_i$ such that
    \begin{itemize}
        \item $\sum_{j \in \bigcup^i_{\ell=1} E_\ell} p_j =x$, and
        \item $\sum_{j \in E_\ell, \ell \leq i_0} p_j \leq d^{(i_0)}$ holds for each $i_0 \in \{1,\ldots,i\}$.
    \end{itemize}
\end{lemma}

\begin{proof}
    The proof is by induction on $i$. For $i=1$, note that $S_1 = \SSS(X_1) \setminus \{ x : x > d^{(1)}\}$ at the end of step $5(1)$. Since $\SSS(X_1)$ includes the total processing time of any subset of jobs in $J_1$, the first condition of the lemma holds. Since $\{ x : x > d^{(1)}\}$ includes all integers violating the second condition of the lemma, the second condition holds.

    Let $i > 1$, and assume the lemma holds for $i-1$. Consider some $x \in S_i$ at the end of the second step of $5(i)$. Then by Definition~\ref{def:sumset}, we have $x=x_1+x_2$ for some $x_1 \in S_{i-1}$ and $x_2 \in \SSS(X_i)$ due the first step of $5(i)$. By definition of $\SSS(X_i)$, there is some $E_i \subseteq J_i$ with total processing time $x_2$. By our inductive hypothesis there is $E_1 \subseteq J_1,\ldots,E_{i-1} \subseteq J_{i-1}$ such that $\sum_{j \in \bigcup^i_{\ell=1} E_\ell} p_j =x_1$, and $\sum_{j \in E_\ell, \ell \leq i_0} p_j \leq d^{(i_0)}$ holds for each $i_0 \in \{1,\ldots,i-1\}$. Furthermore, by the second step of $5(i)$, we know that $\sum_{j \in E_\ell, \ell \leq i} p_j = x \leq d^{(i)}$. Thus, $E_1,\ldots,E_i$ satisfy both conditions of the lemma.
\end{proof}

Let us next analyze the time complexity of the \textsc{SumsetScheduler} algorithm. Steps~1 and~2 can be both performed in $\tilde{O}(n)=\tilde{O}(P)$ time. Next observe that step 3 can be done in total $\tilde{O}(P)$ time using Lemma~\ref{lem:SubsetSumSet}, as $X_2,\ldots,X_{D_{\#}}$ is a partition of the set of all processing times of $J$, and these all sum up to $P$. Next, according to Lemma~\ref{lem:sumset}, each sumset operation at step 5 can be done in time proportional to the largest element in the two sets, which is always at most $P$. Thus, since we perform at most $D_{\#}$ sumset operations, the merging step requires $\tilde{O}(D_{\#} \cdot P)$ time, which gives us the total running time of the algorithm above.

Another way to analyze the running time of \textsc{SumsetScheduler} is to observe that the maximum element participating in the $i$th sumset is bounded by $d^{(i+1)}$. It follows that we can write the running time of the merging step as $\tilde{O}(D)$, where $D = \sum^{D_{\#}}_{i=1} d^{(i)}$. Thus, we have just shown that $1||\sum p_jU_j$ can be solved in $\tilde{O}(\min\{D_{\#} \cdot P, D + P\})$ time, completing the proof of Theorem~\ref{thm:SumsetScheduler}.


\section{Algorithm via Fast Skewed Convolutions}
\label{section:ConvScheduler}%

We next present our $\tilde{O}(P^{7/4})$ time algorithm for $1||\sum p_jU_j$, providing a proof of Theorem~\ref{thm:ConvScheduler}. As in the previous section, we let $d^{(1)}<\cdots<d^{(D_\#)}$ denote the $D_\# \leq n$ different due dates of the input jobs $J$, and $J_i=\{j : d_j = d^{(i)}\}$ and $X_i = \{ p_j : j \in J_i \}$ as in Section~\ref{section:SumsetScheduler} for each $i \in \{1,\ldots,D_{\#}\}$.

For a consecutive subset of indices $I=\{i_0,i_0+1,\ldots,i_1\}$, with $i_0,\ldots,i_1 \in \{1,\ldots,D_\#\}$, we define a vector $M(I)$, where $M(I)[x]$ equals the latest (that is, maximum) time point~$x_0$ for which there is a subset of the jobs in $\bigcup_{i \in I} J_i$ with total processing time equal to $x$ that can all be scheduled early in an EDD schedule starting at $x_0$. If no such subset of jobs exists, we define $M(I)[x]=+\infty$.

For a singleton set $I=\{i\}$, the vector $M(I)$ is easy to compute once we have computed the set $\SSS(X_i)$:
\begin{equation}
    \label{eqn:singltons}%
    M(\{i\})[x]=
    \begin{cases}
        d^{(i)} - x & \text{if $x \in \SSS(X_i)$ and $x \leq d^{(i)}$},\\
        +\infty & \text{otherwise}.
    \end{cases}
\end{equation}
For larger sets of indices, we have the following lemma.
\begin{lemma}
    \label{lem:skewed2scheduling}%
    Let $I_1=\{i_0,i_0+1,\ldots,i_1\}$ and $I_2=\{i_1+1,i_1+2,\ldots,i_2\}$ be any two sets of consecutive indices with $i_0,\ldots,i_1, \ldots, i_2  \in \{1,\ldots,D_\#\}$. Then for any value $x$ we have:
    \[
    M(I_1 \cup I_2)[x]=\max_{x_1+x_2=x} \min \{M(I_1)[x_1], M(I_2)[x_2] - x_1\}.
\]
\end{lemma}

\begin{proof}
    Let $I = I_1 \cup I_2$. Then $M(I)[x]$ is the latest time point after which a subset
    of jobs $J^* \subseteq \bigcup_{i \in I} J_i$ of total processing time $x$ can be
    scheduled early in an EDD schedule. Let $x_1$ and $x_2$ be the total processing times
    of jobs in $J^*_1 = J^*  \cap \left( \bigcup_{i \in I_1} J_i \right)$ and $J^*_2 = J^*
    \cap \left( \bigcup_{i \in I_2} J_i \right)$, respectively. Then $x= x_1 + x_2$.
    Clearly, $M(I)[x] \leq M(I_1)[x_1]$, since we have to start scheduling the jobs in
    $J^*_1$ at time $M(I_1)[x_1]$ by latest. Similarly, it holds that
    ${M(I)[x] \leq M(I_2)[x_2]-x_1}$ since the jobs in $J^*_2$ are scheduled at latest at
    $M(I_2)[x_2]$
    and the jobs in $J^*_1$ have to be processed before that time point in an EDD
    schedule.
    In combination, we have shown that LHS $\leq$ RHS in the equation of the lemma.

    To prove that LHS $\geq$ RHS, we construct a feasible schedule for jobs in $\bigcup_{i \in I} J_i$ starting at RHS. Let $x_1$ and $x_2$ be the two values with $x_1 + x_2 = x$ that maximize RHS. Then there is a schedule which schedules some jobs $J^*_1 \subseteq  \bigcup_{i \in I_1} J_i$ of total processing time $x_1$ beginning at time $\min \{M(I_1)[x_1], M(I_2)[x_2] - x_1\} \leq M(I_1)[x_1]$, followed by a another subset of jobs $J^*_2  \subseteq \bigcup_{i \in I_2} J_i$ of total processing time $x_2$ starting at time $\min \{M(I_1)[x_1], M(I_2)[x_2] - x_1\}+x_1 \leq M(I_2)[x_2]$. This is a feasible schedule starting at time RHS for a subset of jobs in $\bigcup_{i \in I} J_i$ which has total processing time~$x$.
\end{proof}

Note that the equation given in Lemma~\ref{lem:skewed2scheduling} is close but not precisely the equation defined in Definition~\ref{def:SkewedConv} for the $(\min,\max)$-Skewed-Convolution problem. Nevertheless, the next lemma shows that we can easily translate between these two concepts.

\begin{lemma}
    Let $A$ and $B$ be two integer vectors of $P$ entries each. Given an algorithm for computing the $(\max,\min)$-Skewed-Convolution of $A$ and $B$ in $T(P)$ time, we can compute in $T(P)+O(P)$ time the vector $C = A \otimes B$ defined by
\[
    C[x] =\max_{x_1+x_2=x} \min \{A[x_1],B[x_2]-x_1\}.
\]
\end{lemma}

\begin{proof}

    Given $A$ and $B$, construct two auxiliary vectors $A_0$ and $B_0$ defined by $A_0[x]=B[x]+x$ and $B_0[x]=A[x]$ for each entry~$x$. Compute the $(\max,\min)$-Skewed-Convolution of $A_0$ and $B_0$, and let $C_0$ denote the resulting vector. We claim that the vector $C$ defined by $C[x]=C_0[x]-x$ equals $A \otimes B$. Indeed, we have
    \begin{eqnarray*}
        C_0[x]-x &=& \max_{x_1+x_2=x}\min\{A_0[x_1],B_0[x_2]+x\}-x\\
                 &=& \max_{x_1+x_2=x}\min\{A_0[x_1]-x,B_0[x_2]\}\\
                 &=& \max_{x_1+x_2=x}\min\{B[x_1]+x_1-x,A[x_2]\}\\
                 &=& \max_{x_1+x_2=x}\min\{B[x_1]-x_2,A[x_2]\}\\
                 &=& \max_{x_1+x_2=x}\min\{A[x_1],B[x_2]-x_1\},
    \end{eqnarray*}
    where in the third step we expanded the definition of $A_0$ and $B_0$ and in the last step we used the symmetry of $x_1$ and $x_2$.
\end{proof}

We are now in position to describe our algorithm called \textsc{ConvScheduler} which is depicted in Algorithm~\ref{alg:conv-scheduler}.
The algorithm first computes the subset sums $\SSS(X_1), \ldots, \SSS(X_{D_\#})$, and the set of vectors $\mathcal{M} = \{M_1, \ldots, M_{D_\#}\}$. Following this, it iteratively combines every two consecutive vectors in $\mathcal{M}$ by using the $\otimes$ operation. The algorithm terminates when $\mathcal{M}=\{M_1\}$, where at this stage $M_1$ corresponds to the entire set of input jobs $J$. It then returns $P-x$, where $x$ is the maximum value with $M_1[x] < \infty$; by definition, this corresponds to a schedule for $J$ with $P-x$ total processing time of tardy jobs. For convenience of presentation, we assume that $D_\#$ is a power of 2.

\begin{algorithm}[h]
    \caption{\textsc{ConvScheduler($J$)}} \label{alg:conv-scheduler}
    \begin{algorithmic}[1]
        \State Let $d^{(1)} < \ldots < d^{(D_{\#})}$ denote the different due dates of jobs in $J$.
        \State Compute $X_i = \{p_j : d_j = d^{(i)}\}$ for each $i \in \{1,\ldots,D_\#\}$.
        \State Compute $\SSS(X_1), \ldots, \SSS(X_{D_\#})$.
        \State Compute $\mathcal{M}=\{M_1=M(1),\ldots,M_{D_\#}=M(D_\#)\}$.
        \While{$|\mathcal{M}| > 1$}
        \Statex\textbf{\,\,--\,\,\,}Compute $M_i = M_{2i-1} \otimes M_{2i}$ for each $i \in \{1, \ldots, |\mathcal{M}|/2 \}$.
        \EndWhile
        \vspace{-.2ex}
        \State Return $P-x$, where $x$ is the maximum value with $M_1[x] < \infty$.
    \end{algorithmic}
\end{algorithm}

Correctness of this algorithm follows directly from Lemma~\ref{lem:skewed2scheduling}. To analyze its time complexity, observe that steps 1--4 can be done in $\tilde{O}(P)$ time (using Lemma~\ref{lem:SubsetSumSet}). Step 5 is performed $O(\log D_\#)=O(\log P)$ times, and each step requires a total of $\tilde{O}(P^{7/4})$ time according to Theorem~\ref{thm:SkewedConv}, as the total sizes of all vectors at each step is $O(P)$. Finally, step~6 requires $O(P)$ time. Summing up, this gives us a total running time of $\tilde{O}(P^{7/4})$, and completes the proof of Theorem~\ref{thm:ConvScheduler} (apart from the proof of Theorem~\ref{thm:SkewedConv}).

\section{Fast Skewed Convolutions}
\label{section:SkewedConv}%

In the following section we present our algorithm for $(\max,\min)$-Skewed-Convolution, and provide a proof for Theorem~\ref{thm:SkewedConv}. Let $A = (A[i])^n_{i=0}$ and $B = (B[j])^n_{j=0}$ denote the input vectors for the problem throughout the section.

We begin by first defining the problem slightly more generally, in order to facilitate our recursive strategy later on. For this, for each integer $\ell \in \{0,\ldots,\log n\}$, let $A^\ell=\lfloor A/2^\ell \rfloor$ and $B^\ell=\lfloor B/2^\ell \rfloor$, where rounding is done component-wise. We will compute vectors $C^\ell = (C^\ell[k])^{2n}_{k=0}$ defined by:
\[
C^\ell[k] =\max_{i+j=k} \min \{A^\ell[i],B^\ell[j]+\lfloor k/2^\ell \rfloor \}.
\]
Observe that a solution for $\ell = 0$ yields a solution to the original $(\max,\min)$-Skewed-Convolution problem, and for $\ell \geq \log 2n$ the problem degenerates to $(\max,\min)$-Convolution.

We next define a particular kind of additive approximation of vectors $C^\ell$. We say that a vector $D^\ell$ is a \emph{good approximation} of $C^\ell$ if $C^\ell[k]-2 \leq D^\ell[k] \leq C^\ell[k]$ for each $k \in \{0,\ldots,2n\}$. Now, the main technical part of our algorithm is encapsulated in the following lemma.
\begin{lemma}
    \label{lem:ApproxSkewedConv}%
    There is an algorithm that computes $C^\ell$ in $\tilde{O}(n^{7/4})$ time, given $A^\ell$, $B^\ell$, and a good approximation $D^\ell$ of $C^\ell$.
\end{lemma}

We postpone the proof of Lemma~\ref{lem:ApproxSkewedConv} for now, and instead show that it directly yields our desired algorithm for $(\max,\min)$-Skewed-Convolution:

\begin{proof}[of Theorem~\ref{thm:SkewedConv}]
    In order to compute $C = C^0$, we perform an (inverse) induction on~$\ell$: As mentioned before, if $\ell \geq \log 2n$, then we can neglect the ``+ $\lfloor k/2^\ell \rfloor$'' term and compute $C^\ell$ in $\tilde{O}(n^{3/2})=\tilde{O}(n^{7/4})$ time using a single $(\max,\min)$-Convolution computation~\cite{Kosaraju1989}.

    For the inductive step, let $\ell < \log 2n$ and assume that we have already computed
    $C^{\ell+1}$. We construct the vector $D^{\ell}=2C^{\ell+1}$, and argue that it is a
    good approximation of $C^\ell$. Indeed, for each entry $k$, on the one hand, we have:
    \begin{eqnarray*}
        D^{\ell}[k] &=& 2C^{\ell+1}[k] = 2 \cdot \max_{i+j=k} \min \{\lfloor A^\ell[i]/2 \rfloor, \lfloor B^\ell[j]/2 \rfloor +\lfloor k/2^{\ell+1}   \rfloor \}\\ &\leq& \max_{i+j=k} \min \{A^\ell[i],B^\ell[j]+\lfloor k/2^\ell \rfloor \} = C^\ell[k];
    \end{eqnarray*}
    and on the other hand, we have:
    \begin{eqnarray*}
        D^{\ell}[k] &=& 2C^{\ell+1}[k] = 2 \cdot \max_{i+j=k} \min \{\lfloor A^\ell[i]/2 \rfloor, \lfloor B^\ell[j]/2 \rfloor +\lfloor k/2^{\ell+1}   \rfloor \} \\ &\geq& \max_{i+j=k} \min \{A^\ell[i]-1,B^\ell[j]+\lfloor k/2^\ell \rfloor -2\} \geq C^\ell[k] -2.
    \end{eqnarray*}
    Thus, using $D^{\ell}$ we can apply Lemma~\ref{lem:ApproxSkewedConv} above to obtain $C^{\ell}$ in $\tilde{O}(n^{7/4})$ time. Since there are $O(\log n)$ inductive steps overall, this is also the overall time complexity of the algorithm.
\end{proof}

It remains to prove Lemma~\ref{lem:ApproxSkewedConv}. Recall that we are given $A^\ell$, $B^\ell$, and $D^\ell$, and our goal is to compute the vector $C^\ell$ in $\tilde{O}(n^{7/4})$ time. We construct two vectors $L^\ell$ and $R^\ell$ with $2n$ entries each, defined by\[
L^\ell[k] = \max\left\{A^\ell[i_0] : \text{\parbox{5.4cm}{\centering$A^\ell[i_0] \leq B^\ell[k-i_0] + \lfloor k/2^\ell \rfloor$ and\\[.3ex] $D^\ell[k] \leq A^\ell[i_0] \leq D^\ell[k] + 2$}} \right\},
\] and \[
R^\ell[k] = \max\left\{
    B^\ell[j_0] + \lfloor k/2^\ell \rfloor : \text{\parbox{6.2cm}{\centering$B^\ell[j_0]+\lfloor k/2^\ell \rfloor \leq A^\ell[k-j_0]$ and\\[.3ex]$D^\ell[k] \leq B^\ell[j_0] + \lfloor k/2^\ell \rfloor \leq D^\ell[k] + 2$}} \right\}
\]
for $k\in\{0,\ldots,2n\}$. That is, $L^\ell[k]$ and $R^\ell[k]$
respectively capture the largest value attained as the left-hand side or right-hand side of the inner min-operation in $C^\ell[k]$, as long as that value lies in the feasible region approximated by $D^\ell[k]$. Since $D^\ell$ is a good approximation, the following lemma is immediate from the definitions:
\begin{lemma}
    \label{lem:LandR}%
    $C^\ell[k]=\max\{L^\ell[k],R^\ell[k]\}$ for each $k\in\{0,\ldots,2n\}$.
\end{lemma}




According to Lemma~\ref{lem:LandR}, it suffices to compute $L^\ell$ and $R^\ell$. We focus on computing $L^\ell$ as the algorithm for computing $R^\ell$ follows after applying minor modifications.

Let $0 < \delta < 1$ be a fixed constant to be determined later.  We say that an index $k \in \{0,\ldots,n\}$ is \emph{light} if
\[
|\{i : D^\ell[k] \leq A^\ell[i] \leq D^\ell[k]+2\}| \leq n^\delta.
\]
Informally, $k$ is light if the number of candidate entries $A^\ell[i]$ which can equal $C^\ell[k]$ is relatively small (recall that $D^\ell[k] \leq C^\ell[k] \leq D^\ell[k] +2$, as $D^\ell$ is a good approximation of~$C^\ell$). If $k$ is not light then we say that it is \emph{heavy}.

Our algorithm for computing $L^\ell$ proceeds in three main steps: In the first step it handles all light indices, in the second step it sparsifies the input vector, and in the third step it handles all heavy indices:
\begin{itemize}
    \item \emph{Light indices:} We begin by iterating over all light indices $k \in \{0,\ldots,2n\}$. For each light index $k$, we iterate over all entries $A^\ell[i]$ satisfying $D^\ell[k] \leq A^\ell[i] \leq D^\ell[k]+2$, and set $L^\ell[k]$ to be the maximum $A^\ell[i]$ among those entries with $A^\ell[i] \leq B^\ell[k-i] + \lfloor k/2^\ell \rfloor$.
        Note that after this step, we have \[
        L^\ell[k]=\max\{A^\ell[i_0] : A^\ell[i_0] \leq B^\ell[k-i_0] + \lfloor k/2^\ell
    \rfloor \text{ and } D^\ell[k] \leq A^\ell[i_0] \leq D^\ell[k] + 2\}\] for each light index $k$.

        \bigskip
    \item \emph{Sparsification step:} After dealing with the light indices, several entries of~$A^\ell$ become redundant. Consider an entry $A^\ell[i]$ for which $|\{i_0 : A^\ell[i]-2 \leq A^\ell[i_0] \leq A^\ell[i] +2 \}| \leq n^\delta$. Then all indices $k$ for which $L^\ell[k]$ might equal $A^\ell[i]$ must be light, and are therefore already dealt with in the previous step. Consequently, it is safe to replace $A^\ell[i]$ by~$-\infty$ so that $A^\ell[i]$ no longer plays a role in the remaining computation.

        \bigskip
    \item \emph{Heavy indices:} After the sparsification step $A^\ell$ contains few distinct values. Thus, our approach is to fix any such value $v$ and detect whether $L^\ell[k] \geq v$. To that end, we translate the problem into an instance of $(\max,\min)$-Convolution: Let $(A^\ell_v[i])^n_{i=0}$ be an be an indicator-like vector defined by $A^\ell_v[i] = +\infty$ if $A^\ell[i]=v$, and otherwise $A^\ell_v[i] = -\infty$. We next compute the vector $L^\ell_v$ defined by
        $L^\ell_v[k]=\lfloor k/2^\ell \rfloor + \max_{i+j=k} \min\{A^\ell_v[i],B^\ell[j]\}$ using a single computation of $(\max,\min)$-Convolution.

        \hspace{15pt} We choose \[
    L^\ell[k] = \max\{v : \text{$L^\ell_v[k] \geq v$ and $D^\ell[k] \leq v \leq D^\ell[k] + 2$}\}\]
        for any heavy index $k$ and claim that $L^\ell[k]$ equals $\max\{A^\ell[i_0] : A^\ell[i_0] \leq B^\ell[k-i_0] + \lfloor k/2^\ell \rfloor\}$.
        On the one hand, if $L^\ell_v[k] \geq v$ then there are indices $i$ and $j$ with $i+j=k$ for which $A^\ell[i]=v$ and $B^\ell[j] + \lfloor k/2^\ell \rfloor \geq A^\ell[i]=v$. Thus, the computed value $L^\ell[k]$ is not greater than
        \[ L^\ell[k] \le \max\{A^\ell[i_0] : \text{$A^\ell[i_0] \leq B^\ell[k-i_0] + \lfloor k/2^\ell \rfloor$ and $D^\ell[k] \leq A^\ell[i_0] \leq D^\ell[k] + 2$}\}.\]
        On the other hand, for all values $v$ for which $A^\ell[i]=v$ for some $i \in \{0,\ldots,n\}$, we have if $v = A^\ell[i] \leq B^\ell[k-i] + \lfloor k/2^\ell \rfloor$ then $A^\ell_v[i] = -\infty$, which in turn implies that $A^\ell_v[i] \geq B^\ell[k-i] + \lfloor k/2^\ell \rfloor \geq A^\ell[i] = v$.
        Thus, our selection of $L^\ell[k]$ is also at least as large as
        \[L^\ell[k] \geq \max\{A^\ell[i_0] : \text{$A^\ell[i_0] \leq B^\ell[k-i_0] + \lfloor k/2^\ell \rfloor$ and $D^\ell[k] \leq A^\ell[i_0] \leq D^\ell[k] + 2$}\},
        \] and hence, these two values must be equal.
\end{itemize}

This completes the description of our algorithm. As we argued its correctness above, what remains is to analyze its time complexity. Note that we can determine in $O(\log n)$ time whether an index $k$ is light or heavy, by first sorting the values in $A^\ell$. For each light index $k$, determining $L^\ell[k]$ can be done in $O(n^\delta)$ time (on the sorted $A^\ell$), giving us a total of $\tilde{O}(n^{1+\delta})$ time for the first step. For the second step, we can determine whether a given entry $A^\ell[i]$ can be replaced with $-\infty$ in $O(\log n)$ time, giving us a total of $\tilde{O}(n)$ time for this step.

Consider then the final step of the algorithm. Observe that after exhausting the sparsification step, $A^\ell$ contains at most $O(n^{1-\delta})$ many distinct values: For any surviving value~$v$, there is another (perhaps different) value $v'$ of difference at most 2 from $v$ that occurs at least $n^\delta$ times in~$A^\ell$, and so there can only be at most $O(n^{1-\delta})$ such distinct  values. Thus, the running time of this step is dominated by the running time of $O(n^{1-\delta})$ $(\max,\min)$-Convolution computations, each requiring $\tilde{O}(n^{3/2})$ time using the algorithm of~\cite{Kosaraju1989}, giving us a total of $\tilde{O}(n^{5/2-\delta})$ time for this step.

Thus, the running time of our algorithm is dominated by the $\tilde{O}(n^{1+\delta})$ running time of its first step, and the $\tilde{O}(n^{5/2-\delta})$ running time of its last step. Choosing $\delta=3/4$ gives us $\tilde{O}(n^{7/4})$ time for both steps, which is the time promised by Lemma~\ref{lem:ApproxSkewedConv}. Thus, Lemma~\ref{lem:ApproxSkewedConv} holds.

\section{Discussion and Open Problems}
\label{section:discussion}%

In this paper we presented two algorithms for the $1||\sum p_jU_j$ problem; the first running in $\tilde{O}(P^{7/4})$ time, and the second running in $\tilde{O}(\min\{P \cdot D_{\#}, P + D\})$ time. Both algorithms provide the first improvements over the classical Lawler and Moore algorithm in 50 years, and use more sophisticated tools such as polynomial multiplication and fast convolutions. Moreover, both algorithms are very easy to implement given a standard ready made FFT implementation for fast polynomial multiplication. Nevertheless, there are still a few ways which our results can be improved or extended:
\begin{itemize}
    \item \emph{Multiple machines:} A natural extension of the $1||\sum p_jU_j$ problem is to the setting of multiple parallel machines, the $Pm||\sum p_jU_j$. Lawler and Moore's algorithm can be used to solve $Pm||\sum p_jU_j$ in $O(P^m \cdot n)$ time, where $m$ is the number of machines. A priori, there is no reason to believe that this cannot be improved to $\tilde{O}(P^m)$, or even better. It is not hard to extend the algorithm in  Theorem~\ref{thm:SumsetScheduler} to an algorithm with running time $\tilde{O}(P^m \cdot D_{\#})$ for the $m$ parallel machine setting, by using $m$-variate polynomials for implementing sumsets in Lemma~\ref{lem:sumset}. However, a similar extension for the algorithm in Theorem~\ref{thm:ConvScheduler} is far less direct.

        \bigskip
    \item \emph{Even faster skewed convolutions:} We have no indication that our algorithm for $(\max,\min)$-Skewed-Convolution is the fastest possible. It would interesting to see whether one can improve its time complexity, say to $\tilde{O}(P^{3/2})$. Naturally, any such improvement would directly improve Theorem~\ref{thm:ConvScheduler}.

        \hspace{15pt} Conversely, one could try to obtain some sort of lower bound for the problem, possibly in the same vein as Theorem~\ref{thm:Known LB1}. Improving the time complexity beyond $\tilde O(P^{3/2})$ seems difficult as this would directly imply an improvement to the $(\max,\min)$-Convolution problem. Indeed, let $A$, $B$ be a given $(\max,\min)$-Convolution instance and construct vectors $A_0$, $B_0$ with $A_0[i] = N \cdot A[i]$ and $B_0[j] = N \cdot B[j]$ for $N = 2n+1$. If $C_0$ is the $(\max,\min)$-Skewed-Convolution of $A_0$ and $B_0$ (that is, $C_0[k] = \max_{i + j = k} \min\{A_0[i], B_0[j] + k\}$), then the vector $C$ with $C[k] = \lfloor C_0[k] / N \rfloor$ is the $(\max,\min)$-Convolution of $A$ and $B$.

        \bigskip
    \item \emph{Other scheduling problems:} Finally, it will be interesting to see other scheduling problems where the techniques used in this paper can be applied. A good first place to start might be to look at other problems which directly generalize Subset Sum.
\end{itemize}

\bibliographystyle{plain}
\bibliography{biblo}

\end{document}